\documentclass[11pt]{article}

\usepackage{fullpage}
\usepackage{float}
\usepackage{epsfig}
\usepackage{amsmath}
\usepackage{amssymb}
\usepackage{amstext}
\usepackage{amsthm}
\usepackage{xspace}
\usepackage{latexsym}
\usepackage{verbatim}
\usepackage{multirow}
\usepackage{ifthen}
\usepackage{url}
\usepackage{hyperref}
\usepackage{natbib}

\usepackage[ruled,vlined]{algorithm2e}

\DeclareMathAlphabet{\mathitbf}{OML}{cmm}{b}{it}

\newtheorem{theorem}{Theorem}[section]
\newtheorem{definition}{Definition}

\newtheorem{lemma}[theorem]{Lemma}

\newtheorem{assumption}{Assumption}
\renewcommand{\comment}[1]{}

\newenvironment{proofof}[1]{\vspace{0.1in}\noindent{\sc Proof of #1.}}{\hfill\qed}

\newcommand{\prob}[2][]{\text{\bf Pr}\ifthenelse{\not\equal{}{#1}}{_{#1}}{}\!\left[#2\right]}
\newcommand{\expect}[2][]{\text{\bf E}\ifthenelse{\not\equal{}{#1}}{_{#1}}{}\!\left[#2\right]}

\newcommand{\yestag}{\addtocounter{equation}{1}\tag{\theequation}}

\newcommand{\argmax}{\operatorname{argmax}}

\newcommand{\agind}[1][i]{_{#1}}

\newcommand{\inversed}[1]{#1^{-1}}
\newcommand{\ironed}{\bar}

\newcommand{\differentiated}[1]{#1'}
\newcommand{\fortype}{\tilde}
\newcommand{\estimated}{\hat}
\newcommand{\sampled}{\hat}

\newcommand{\noaccents}[1]{#1}
\newcommand{\composed}[3]{#1{#2{#3}}}

\newcommand{\newindexedvar}[4][\noaccents]{%
\expandafter\newcommand\expandafter{\csname #2\endcsname}{#1{#4}}%
\expandafter\newcommand\expandafter{\csname #2s\endcsname}{#1{\boldsymbol{#4}}}%
\expandafter\newcommand\expandafter{\csname #2sm#3\endcsname}[1][#3]{#1{\boldsymbol{#4}}_{-##1}}%
\expandafter\newcommand\expandafter{\csname #2#3\endcsname}[1][#3]{#1{#4}\agind[##1]}%
\expandafter\newcommand\expandafter{\csname #2#3th\endcsname}[1][#3]{#1{#4}_{(##1)}}%
}

\newcommand{\rev}{R}
\newcommand{\irev}{\overline{R}}
\newcommand{\marg}{R'}
\newcommand{\knalloc}[1][k]{x^{(#1:n)}}
\newcommand{\kalloc}[1][k]{x_{#1}}
\newcommand{\quant}{q}

\newindexedvar{wal}{k}{w}
\newindexedvar[\differentiated]{margwal}{k}{\wal}
\newindexedvar{cumwal}{k}{W}

\newindexedvar[\ironed]{iwal}{k}{\wal}
\newindexedvar[\ironed]{cumiwal}{k}{\cumwal}
\newindexedvar[\composed{\differentiated}{\ironed}]{margiwal}{k}{\wal}

\newindexedvar{yal}{k}{y}
\newindexedvar{cumyal}{k}{Y}
\newindexedvar[\ironed]{iyal}{k}{\yal}
\newindexedvar[\composed{\differentiated}{\ironed}]{margiyal}{k}{\yal}

\newindexedvar{murev}{k}{P}
\newindexedvar[\estimated]{emurev}{k}{\murev}
\newindexedvar[\differentiated]{mumarg}{k}{\murev}

\newindexedvar[\ironed]{imurev}{k}{\murev}
\newindexedvar[\composed{\differentiated}{\ironed}]{imumarg}{k}{\murev}

\newcommand{\ghat}{\hat{g}}
\newcommand{\bhat}{\hat{b}}

\newcommand{\vhat}{\hat{v}}
\renewcommand{\[}{\left[}
\renewcommand{\]}{\right]}

\newcommand{\rateN}[1][\samples]{r(#1)}
\newcommand{\rate}{\rateN} 
\DeclareMathOperator{\MSE}{MSE}
\newcommand{\MSEN}[1][]{\MSE_{#1}(\samples)}

\newindexedvar{val}{i}{v}

\newindexedvar{eff}{i}{e}

\newindexedvar{bid}{i}{b}
\newindexedvar[\estimated]{ebid}{i}{\bid}
\newindexedvar[\sampled]{sbid}{i}{\bid}

\newindexedvar{strat}{i}{s}

\newindexedvar{virt}{i}{\phi}
\newindexedvar[\inversed]{virtinv}{i}{\virt}

\newindexedvar[\ironed]{ivirt}{i}{\virt}

\newindexedvar{dist}{i}{F}

\newindexedvar{dens}{i}{f}

\newindexedvar{price}{i}{p}
\newindexedvar[\fortype]{tprice}{i}{p}

\newindexedvar{alloc}{i}{x}

\newindexedvar[\fortype]{talloc}{i}{\alloc}

\newindexedvar{util}{i}{u}
\newindexedvar[\fortype]{tutil}{i}{u}

\newcommand{\REV}[1]{\mathbf{Rev}[#1]}

\newcommand{\samples}{N}
\newcommand{\eps}{\epsilon}

\newcommand{\opt}[2]{\mathbf{OPT}(#1,#2)}

\begin{document}
\title{Mechanism Design for Data Science}
\author{Shuchi Chawla\thanks{University of Wisconsin - Madison}
\and
Jason Hartline\thanks{Northwestern University}
\and
Denis Nekipelov\thanks{University of California - Berkeley and University of Virginia}
}

\maketitle

\begin{abstract}
Good economic mechanisms depend on the preferences of participants in
the mechanism.  For example, the revenue-optimal auction for selling
an item is parameterized by a reserve price, and the appropriate
reserve price depends on how much the bidders are willing to pay.  A
mechanism designer can potentially learn about the participants'
preferences by observing historical data from the mechanism; the
designer could then update the mechanism in response to learned
preferences to improve its performance.  The challenge of such an
approach is that the data corresponds to the actions of the
participants and not their preferences.  Preferences can potentially
be inferred from actions but the degree of inference possible depends
on the mechanism.  In the optimal auction example, it is impossible to
learn anything about preferences of bidders who are not willing to pay
the reserve price.  These bidders will not cast bids in the auction
and, from historical bid data, the auctioneer could never learn that
lowering the reserve price would give a higher revenue (even if it
would).  To address this impossibility, the auctioneer could sacrifice
revenue optimality in the initial auction to obtain better inference
properties so that the auction's parameters can be adapted to changing
preferences in the future.  This paper develops the theory for optimal
mechanism design subject to good inferability.
\end{abstract}



\section{Introduction}
\label{s:intro}


The promise of data science is that if data from a system can be
recorded and understood then this understanding can potentially be
utilized to improve the system.  Behavioral and economic data,
however, is different from scientific data in that it is subjective to
the system.  Behavior changes with system changes, and to predict
behavior for any given system change or to optimize over system
changes, the model that generates the behavior must be inferred from
the behavior.  The ease with which this inference can be performed
generally also depends on the system.  Trivially, a system that
ignores behavior does not admit any inference of a behavior
generating model that can be used to predict behavior in a system that
is responsive to behavior.  To realize the promise of data science in
economic systems, a theory for the design of such systems must also
incorporate the desired inference properties.


Consider as an example the revenue-maximizing auctioneer.  If the
auctioneer has knowledge of the distribution of bidder values then she
can run the first-price auction with a reserve price that is tuned to
the distribution.  Under some mild distributional assumptions, with
the appropriate reserve price the first-price auction is revenue
optimal \citep{mye-81}.  Notice that the historical bid data for the
first-price auction with a reserve price will never have bids for
bidders whose values are below the reserve.  Therefore, there is no
data analysis that the auctioneer can perform that will enable
properties of the distribution of bidder values below the reserve
price to be inferred. It could be, nonetheless, that over time the
population of potential bidders evolves and the optimal reserve price
lowers.  This change could go completely unnoticed in the auctioneer's
data.  The two main tools for optimizing revenue in an auction are
reserve prices (as above) and ironing.\footnote{Optimal auctions are
  constructed by mapping values to virtual values and then optimizing
  the virtual surplus (i.e., the sum of the the virtual values of
  agents who are served).  Reserve prices exclude bidders with
  negative virtual values; bidders with value below the reserve price
  may as well not show up at the auction.  Ironing comes from the
  process of constructing virtual values where a non-monotone virtual
  value function is ``ironed'' to be flat.  When two or more bidders
  have the same virtual value they receive the same service
  probability (and payment); bidders in an ironed interval may as well
  make the same bid.}  Both of these tools cause pooling behavior
(i.e., bidders with distinct values take the same action) and
economic inference cannot thereafter differentiate these pooled
bidders.  In order to maintain the distributional knowledge necessary
to be able to run a good auction in the long term, the auctioneer must
sacrifice the short-term revenue by running a non-revenue-optimal
auction.


Economic inference in auctions is based on a very straightforward
premise.  We, the analyst, would like to infer the values of bidders
from their bids.  This is possible because a bidder's bid is in best
response to the bid distribution and the bid distribution is observed
in the data.  Given any value for a bidder, and the empirical
distribution of bids, we can solve for the bidder's utility maximizing
bid. The resulting bid function can be easily obtained, and bids can
be mapped to ranges of values via the inverse of this function.  If
the true equilibrium bid distribution is continuous with strictly
positive density then both the noise in estimation and the width of
the estimated value intervals vanish as the size of the observed data
set increases.  Notice that our earlier observation that nothing can
be learned about the value distribution below the reserve price (or in
an ironed interval) corresponds to the probability that a bidder is
served being constant for bids within a given interval; The bid
distribution would have zero density on this interval.


Our examples above that demonstrate a failure to perform economic
inference are extreme in that the failure was due to total
unresponsiveness of the mechanism.  Intuitively, the degree to which
good inference is possible should depend on the degree of
responsiveness.  The ideal econometric model for inference, however,
does not separate mechanisms by the degree of responsiveness.  As an
example, consider auctioning a single item to one of 100 agents by the
highest-bid-wins auction with either first-price or all-pay semantics.
With first-price semantics only the winner pays his bid; with all-pay
semantics all bidders pay their bids.  Consider a bidder with value
close to the median of the distribution.  With the first-price auction
the bidder bids by shading his value to the expected highest-other-bid
given that his bid is the highest.  For most continuous distributions,
this is slightly below his value.  With the all-pay auction, relative
to the first-price bid, the bidder also reduces his bid by a
multiplicative factor proportional to his likelihood of being the
highest bidder.  For a median value bidder with 99 other bidders, this
probability is $2^{-99}$.  In contrast to the first-price auction
where determining a reasonable bid is relatively easy, it is unlikely
that a bidder would respond in an all-pay auction to a degree of
accuracy that allows inference.  As described above, the classical
model for inference suggests that, in both cases, the value
distribution can be inferred to a precision that vanishes with the
number of observations.  Our theory will enable the distinction of
these two mechanisms as one with good inference and one with bad
inference.


In summary, the goal of this work is in a theory for the design of
mechanisms that perform well in terms of revenue and inference.  The
motivating story that the theory should resolve is that of a
revenue-maximizing auctioneer continuously adapting an auction in a
slowly changing world.  As the distribution of agent preferences
evolves, the equilibrium behavior evolves, and the designer observes
behavior and adapts the mechanism.

\paragraph{Model and Justification}

To explore this question of optimization of revenue of auctions with
good inference we consider the following auction and bidding model.
The auction design space is given by a position auction environment
and either first-price or all-pay semantics.  There are $n$ agents and
$n$ positions and each position has a corresponding service
probability.  The auctioneer may modify the position weights by moving
service probability from high positions to low positions (or
discarding) and then runs a rank-by-bid position auction on the
modified position weights.  The bidders are single-dimensional with
private values drawn independently and identically from a common prior
distribution, have linear utility given by their value for service
received minus any payment made, and bid in Bayes-Nash equilibrium (i.e.,
the model of classical auction theory). Discussion of these modeling choices  follows.

%
%

We choose position auction environments because good revenue can only
be obtained by optimization over the whole range of the distribution
of values; therefore, good inference is especially important.  For
example, reserve pricing alone cannot achieve better than a
logarithmic (in the number of agents) approximation to the optimal
auction \citep{HY-11}.  This choice, then places the problem in an
environment where even obtaining a constant approximation to the
optimal revenue is non-trivial.  Aiming for a constant approximation
will allow us to identify simple mechanisms that provide economic
intuition; whereas, optimal mechanisms may be analytically
intractable and economically opaque.

%
%


We choose to restrict the auctioneer to {\em rank-based} auctions
where the position auction can be modified by shifting service
probability downward (or discarding); importantly we disallow reserve
prices or ironing (by value).  Three reasons follow.  First, as
observed above, reserve prices and ironing do not permit good
inference.  Second, rank-based first-price and all-pay position
auctions have a unique Bayes-Nash equilibrium that is symmetric and in
which bids are always in the same order as values; therefore, the
equilibrium allocation rule is predetermined by the position weights
alone \citep{CH-13}.  Third, as a result that we will prove, both
reserve prices and ironing can be simulated by shifting probability
downward (or discarding) with at most a constant factor loss in
revenue.  Furthermore, the optimal rank-based auction (which obtains
at least the revenue of this simulation) can be obtained by
ironing by rank and discarding low ranks only.  Ironing by rank
considers a set of consecutive position and averages their service
probabilities.  Thus, this restriction is without loss up to a
constant approximation.


We choose first-price and all-pay payment semantics because they are
fundamental auction types with non-truthtelling equilibrium in the
classical model for auction theory.  This choice contrasts with the
standard choice in the mechanism design literature where attention is
often restricted to mechanisms with the truthtelling equilibrium.  Of
course, for truthful mechanisms in the classical model, bids are equal
to values and inference is trivial.  In practice, however, these
nice-in-theory auctions are rarely employed.  \citet{AM-06} give some
explanation for this non-translation from theory to practice in an
essay entitled ``The Lovely but Lonely Vickrey Auction.''  Moreover,
even if an auction possessing a truthtelling equilibrium in theory is
employed in practice, truthtelling is unlikely.  This loss of
truthtelling could come from externalities, exposure, outside options,
or privacy concerns.  In contrast, our restriction to first-price and
all-pay auctions permits the consideration of inference in the
classical, and most fundamental, model of auction theory.

\paragraph{Results and organization}


Our results are as follows.
\begin{itemize}
\item In Section~\ref{s:prelim} we show how to estimate the bid
  distribution from a finite sample of observed bids. Then we provide
  a general principle for inferring values from the estimated
  distribution of bids and its derivative for a non-truthful
  mechanism. We illustrate this principle in application to a first
  price and an all-pay mechanism.

\item We prove in Section~\ref{s:iron-rank} that rank based auctions
  are close to optimal for revenue: for every value distribution and
  every symmetric auction facing a position feasibility
  constraint\footnote{For irregular value distributions, we need the
    assumption that the auction does not iron over the values
    corresponding to the highest $1/n$ quantiles.},
  there exists a rank based auction satisfying the same constraint
  that achieves a $4$-approximation (and in most cases, a close-to-$1$
  approximation). Further, we prove that optimizing for revenue over
  the class of rank based auctions requires knowing $n$ parameters
  of the value distribution that we call the multi-unit revenues.
  These multi-unit revenues are a discrete analog to the revenue
  curve, the derivative of which defines virtual values for
  revenue-optimal auctions.

\item We show how to estimate the multi-unit revenues from samples
drawn from the bid distribution of a mechanism in
Section~\ref{s:param-inf}. We give bounds on the error in estimation
in terms of the allocation rule of the mechanism, and the number of bid
samples. We show that the error decreases as the inverse square root
of the number of bid samples.


\item Also in Section~\ref{s:param-inf}, we show that rank based
auctions achieve a good revenue versus inference tradeoff: for every
$\epsilon>0$, there exists a rank based auction that obtains a
$(1-\epsilon)$ fraction of the optimal rank based revenue, and
achieves error in inferring the multi-unit revenues that scales as
$O(1/\epsilon)$.\footnote{Note that this scaling in error can be
offset by obtaining a large number of bid samples, and by reducing
measurement error in bids.}

\item We show how to estimate the revenue curve for the underlying
value distribution from samples drawn from the bid distribution of a
mechanism in Section~\ref{MWN}. We give bounds on the error in
estimation in terms of the allocation rule of the mechanism, the
number of bid samples, and the measurement error in bids. These bounds
require the allocation function to have a minimum and a maximum slope
at all quantiles. An implication is that the revenue-optimal auction (with
a mostly ``flat'' allocation rule) is generally poor at inference and
requires substantially larger samples to achieve the inference
qualtify of the iron by rank auctions.

We show that the dependence of the estimation error on the
aforementioned parameters is polynomially worse in this setting as
compared to those for the problem of learning just the multi-unit
revenues. 
\end{itemize}


\paragraph{Related Work}

The inference approaches that we use in this paper are related to the
recent work in econometrics literature on the so-called plug-in
estimation of strategic response models.  For instance, \citet{guerre}
analyzes the estimation of the distribution of values from the
distribution of bids in the first-price auction. An overview of
related empirical models applied to other auction environments can be
found in \citet{paarsch}. The principle for inference in these models
is based on the two-step approach discussed in detail in
\citet{bajari}. The idea of the two step method is based on estimation
of the empirical distribution of equilibrium outcome directly from the
data in the first step. In the second step this empirical distribution
is plugged into the first-order condition for each player to obtain
this player's payoff (or value in case of auctions).  This literature,
however, does not consider the question of comparing different
mechanisms in terms of the possibility to infer the parameters of
another mechanism. The mechanism that generates the data is usually
taken as given. Such an analysis has, however, been considered in
non-strategic environments in \citet{chernozhukov}. The non-strategic
structure of the environment makes it significantly easier to analyze,
provided that one does not need to recover the best response
correspondence from the empirical observations, which is one of the
key components of our analysis.

In the mechanism design literature, the problem of designing
mechanisms to enable learning the parameters of a market has not been
considered from a theoretical perspective previously. Several works
have considered the problem of learning optimal pricing schemes in an
online setting (e.g., \citet{BDK+12}). However, these works assume
non-strategic behavior on part of the agents, which makes the
inference much simpler. Other works consider the problem of learning
click-through-rates in the context of a sponsored search auction (a
generalization of the position environment we study) while
simultaneously obtaining good revenue (e.g.,
\citet{DK-09,BSS-09,GLT-12}), however, they restrict attention to
truthful mechanisms, and again do not require inference.

Several works have considered the problem of empirically optimizing
the reserve price of an auction in an online repeated auction
setting (e.g., \citet{Ril-06,BM-09,OS-11}). The most notable of these is the work of
\citet{OS-11}. \citet{OS-11} adapt their mechanism over time to
respond to empirical data by determining the optimal reserve price for
the empirically observed distribution, and then setting a reserve
price that is slightly smaller. This allows for inference around the
optimal reserve price and ensures that the mechanism quickly adapts to
changes in the distribution.

Finally, the theory that we develop for optimizing revenue over the
class of iron by rank auctions is isomorphic to the theory of
envy-free optimal pricing developed by \citet{HY-11}.

\section{Preliminaries}\label{s:prelim}

%
%
%
%
%

\subsection{Auction Theory}

A standard auction design problem is defined by a set $[n] =
\{1,\ldots,n\}$ of $n\ge 2$ agents, each with a private value $\vali$
for receiving a service.  The values are bounded: $\vali\in[0,1]$;
They are independently and identically distributed according to a
continuous distribution $\dist$.  If $\alloci$ indicates the
probability of service and $\pricei$ the expected payment required,
agent $i$ has linear utility $\utili = \vali \alloci - \pricei$.  An
auction elicits bids $\bids = (\bidi[1],\ldots,\bidi[n])$ from the
agents and maps the vector $\bids$ of bids to an allocation
$\tallocs(\bids) = (\talloci[1](\bids),\ldots,\talloci[n](\bids))$,
specifying the probability with which each agent is served, and prices
$\tprices(\bids) = (\talloci[1](\bids),\ldots,\talloci[n](\bids))$,
specifying the expected amount that each agent is required to pay.  An
auction is denoted by $(\tallocs,\tprices)$.

\paragraph{Standard payment formats}

In this paper we study two standard payment formats. In a {\em
  first-price} format, each agent pays his bid upon winning, that is,
$\tpricei(\bids) = \bidi \, \talloci(\bids)$. In an {\em all-pay} format,
each agent pays his bid regardless of whether or not he wins, that is,
$\tpricei(\bids) = \bidi$.

\paragraph{Feasibility in position auction environments}

Auction designers face a feasibility constraint that restricts the
allocations $\allocs$ that the mechanism may produce. In this paper we
focus on position auction environments. In such environments, the
feasibility constraint is given by {\em position weights} $1 \ge
\walk[1]\ge \walk[2] \ge \cdots \ge \walk[n] \geq 0$. An allocation
function assigns agents (randomly) to positions $1$ through $n$, and
an agent assigned to position $i$ gets allocated $\walk[i]$. In other
words, an allocation $\allocs = (\alloci[1],\ldots,\alloci[n])$,
sorted so that $\alloc_1\ge\alloc_2\ge\cdots\ge\alloc_n$, is feasible
if and only if it can be obtained by shifting weight downward (or
discarding), i.e., for all $k\in [n]$, $\sum_{j\le k} \alloci[j] \le
\sum_{j\le k} \walk[j]$. 


\paragraph{Bayes-Nash equilibrium}
The values are independently and identically distributed according to
a continuous distribution $\dist$.  This distribution is common
knowledge to the agents.  A strategy $\strati$ for agent $i$ is a
function that maps the
value of the agent to a bid.  The distribution of values $\dist$ and a
profile of strategies $\strats = (\strati[1],\cdots,\strati[n])$
induces interim allocation and payment rules (as a function of bids)
as follows for agent $i$ with bid $\bidi$.
\begin{align*}
\talloci(\bidi) &= \expect[\valsmi \sim \dist]{\talloci(\bidi,\stratsmi(\valsmi))} \text{ and}\\
\tpricei(\bidi) &= \expect[\valsmi \sim \dist]{\tpricei(\bidi,\stratsmi(\valsmi))}.
\intertext{Agents have linear utility which can be expressed in the interm as:}
\tutili(\vali,\bidi) &= \vali\talloci(\bidi) - \tpricei(\bidi).
\end{align*}
The strategy profile forms a {\em Bayes-Nash equilibrium} (BNE) if for
all agents $i$, values $\vali$, and alternative bids $\bidi$, bidding
$\strati(\vali)$ according to the strategy profile is at least as good
as bidding $\bidi$.  I.e.,
\begin{align}
\label{eq:br}
\tutili(\vali,\strati(\vali)) &\geq \tutili(\vali,\bidi).
\end{align}

A symmetric equilibrium is one where all agents bid by the same
strategy, i.e., $\strats$ statisfies $\strati = \strat$ for some
$\strat$.  For a symmetric equilibrium, the interim allocation and
payment rules are also symmetric, i.e., $\talloci = \talloc$ and
$\strati = \strat$ for all $i$.  For implicit distribution $\dist$ and
symmetric equilibrium given by stratey $\strat$, a mechanism can be
described by the pair $(\talloc,\tprice)$.  \citet{CH-13} show that
the equilibrium of every auction in the class we consider is unique and
symmetric.

The strategy profile allows the mechanism's outcome rules to be
expressed in terms of the agents' values instead of their bids; the
distribution of values allows them to be expressed in terms of the
agents' values relative to the distribution.  This later
representation exposes the geometry of the mechanism.  Define the {\em
  quantile} $\quant$ of an agent with value $\val$ to be the
probability that $\val$ is larger than a random draw from the
distribution $\dist$, i.e., $\quant=\dist(\val)$.  Denote the agent's
value as a function of quantile as $\val(\quant) =
\dist^{-1}(\quant)$, and his bid as a function of quantile as
$\bid(\quant) = \strat(\val(\quant))$.  The outcome rule of the
mechanism in quantile space is the pair
$(\alloc(\quant),\price(\quant)) =
(\talloc(\bid(\quant)),\tprice(\bid(\quant)))$.

\paragraph{Revenue curves and auction revenue}

\citet{mye-81} characterized Bayes-Nash equilibria and this
characteriation enables writing the revenue of a mechanism as a
weighted sum of revenues of single-agent posted pricings.  Formally,
the {\em revenue curve} $\rev(\quant)$ for a given value distribution
specifies the revenue of the single-agent mechanism that serves an
agent with value drawn from that distribution if and only if the
agent's quantile exceeds $\quant$: $\rev(\quant) =
\val(\quant)\,(1-\quant)$. $\rev(0)$ and $\rev(1)$ are defined as
$0$. Myerson's characterization of BNE then implies that the expected
revenue of a mechanism at BNE from an agent facing an allocation rule
$\alloc(\quant)$ can be written as follows:
\begin{eqnarray}
  \label{eq:bne-rev}
  \REV{\alloc} = \expect[\quant]{\rev(\quant)\alloc'(\quant)} = -\expect[\quant]{\rev'(\quant)\alloc(\quant)} 
\end{eqnarray}
where $\alloc'$ and  $\rev'$ denote the derivative of $\alloc$ and $\rev$ with respect to $\quant$, respectively.

The expected revenue of an auction is the sum over the agents of its
per-agent expected revenue; for auctions with symmetric equilibrium
allocaton rule $\alloc$ this revenue is $n \, \REV{\alloc}$.

\paragraph{Rank-based auctions}

In a rank-based auction, the allocation to an agent depends solely on
the rank of his bid among the other agents' bids, and not on the
actual bid. For example, a $k$-highest-bids-win auction is a
rank-based auction, however, a
$k$-highest-bids-win-subject-to-reserve-bid-$r$ auction is not a
rank-based auction.

\subsection{Inference}

The distribution of values, which is unobserved, can be inferred from
the distribution of bids, which is observed.  Once the value
distribution is inferred, other properties of the value distribution
such as its corresponding revenue curve, which is fundamental for
optimizing revenue, can be obtained.  In this section we describe the
basic premise of the inference assuming that the distribution of bids
known exactly.

The key idea behind the inference of the value distribution from the
bid distribution is that the strategy which maps values to bids is a
best response, by equation~\eqref{eq:br}, to the distribution of bids.
As the distribution of bids is observed, and given suitable continuity
assumptions, this best response function can be inverted.

The value distribution can be equivalently specified by distribution
function $\dist(\cdot)$ or value function $\val(\cdot)$; the bid
distribution can similarly be specified by the bid function
$\bid(\cdot)$.  For rank-based auctions (as considered by this paper)
the allocation rule $\alloc(\cdot)$ in quantile space is known
precisely (i.e. it does not depend on the value function
$\val(\cdot)$).  Assume these functions are monotone, continuously
differentiable, and invertible.

\paragraph{Inference for first-price auctions}
Consider a first-price rank-based auction with a symmetric bid
function $\bid(\quant)$ and allocation rule $\alloc(\quant)$ in BNE.
To invert the bid function we solve for the bid that the agent with
any quantile would make.  Continuity of this bid function implies that
its inverse is well defined.  Applying this inverse to the bid
distribution gives the value distribution.

The utility of an agent with quantile $\quant$ as a function of his bid $z$
is
\begin{align*}
\yestag\label{eq:fp-util}
\util(\quant,z) &= (\val(\quant) - z) \, \alloc(\bid^{-1}(z)).\\ 
\intertext{Differentiating with respect to $z$ we get:} 
\tfrac{d}{dz}\util(\quant,z) &= -\alloc(\bid^{-1}(z)) +
\big(\val(\quant) - z\big) \, \alloc'(\bid^{-1}(z))\,
\tfrac{d}{dz}\bid^{-1}(z).\\ 
\intertext{Here $\alloc'$ is the
  derivative of $\alloc$ with respect to the quantile $q$. Because
  $\bid(\cdot)$ is in BNE, the derivative $\tfrac{d}{dz}\util(z,\quant)$ is $0$ at
  $z=\bid(\quant)$. Rarranging, we obtain:}
  \yestag
  \label{eq:fp-inf}
  \val(\quant) &= \bid(\quant) + \tfrac{\alloc(\quant)\,\bid'(\quant)}{\alloc'(\quant)}
\end{align*}

\paragraph{Inference for all-pay auctions}
We repeat the calculation above for rank-based all-pay auctions; the starting
equation \eqref{eq:fp-util} is replaced with the analogous equation for all-pay auctions:
\begin{align*}
\yestag\label{eq:ap-util}
\util(\quant,z) &= \val(\quant)\,\alloc(\bid^{-1}(z)) - z.
\intertext{Differentiating with respect to $z$ we obtain:}
\tfrac{d}{dz}\util(\quant,z) &= \val(\quant)\,\alloc'(\bid^{-1}(z)) \,\frac{d}{dz}\bid^{-1}(z) - 1,\\
\intertext{Again the first-order condition of BNE implies that this expression is zero at $z = \bid(\quant)$; therefore,}
\yestag  \label{eq:ap-inf}
  \val(\quant) & = \tfrac{\bid'(\quant)}{\alloc'(\quant)}.
\end{align*}

\paragraph{Known and observed quantities}
Recall that the functions $\alloc(\quant)$ and $\alloc'(\quant)$ are
known precisely: these are determined by the rank-based auction
definition.  The functions $\bid(\quant)$ and $\bid'(\quant)$ are
observed.  The calculations above hold in the limit as the number of
samples from the bid distribution goes to infinity, at which point these
obserations are precise.


\subsection{Statistical Model and Methods}

In this section we describe the errors in the estimated bid
distribution and standard analyses for rates of convergence.  The
main error in estimation of the bid distribution is the {\em sampling
  error} due to drawing only a finite number of samples from
the bid distribution. 

The analyst obtains $\samples$ samples from the bid distribution. Each
sample is the corresponding agent's best response to the {\em true}
bid distribution. 
We assume that the number of samples $\samples$
is roughly polynomial in $n$, the number of agents in a single auction.

We can estimate the equilibrium bid distribution $\bid(\quant)$ 
as follows.  Let $\sbidi[1],\cdots,\sbidi[\samples]$ denote the
$\samples$ samples drawn from the bid distribution. 
Sort the bids so that $\sbidi[1] \leq \sbidi[2] \leq \cdots \leq
\sbidi[\samples]$ and define the {\em estimated bid distribution}
$\ebid(\cdot)$ as
\begin{align}\label{bid function}
\ebid(\quant) &= \sbidi &\forall i \in \samples, \quant \in [i-1,i)/\samples
\end{align}


\begin{definition}
For function $\bid(\cdot)$ and estimator $\ebid(\cdot)$, the {\em mean
  squared error} as a function of the number of samples $\samples$ is $$\MSEN[\bid] = \expect{ \sup\nolimits_\quant
  \big| \bid(\quant) - \ebid(\quant) \big|^2}^{1/2}.$$
The {\em rate of convergence} of an unbiased estimator, $\rateN$,
captures the dependence of the mean squared error in terms of the
number of samples $\samples$, keeping all other quantities (including,
e.g., $n$) constant, that is, $\rateN \MSEN = \Theta(1)$.
\end{definition}

We will be interested in comparing the error bounds achieved by
different algorithms for inference. Accordingly, we will state these
bounds in terms of the rate of convergence. We will also state
bounds on the mean squared error of the quantities we estimate, in
terms of the mechanism that generates the bids, in order to optimize
for the mechanism. 

\begin{lemma}
\label{error bid function}
The estimator $\ebid(\cdot)$ defined directly from the bids by
equation \eqref{bid function} 
has a rate of convergence of $\rateN = \sqrt N$ and mean squared error
$\MSEN[\bid] = O(\sup_qb'(q)/\sqrt{2N})$. Recalling that $\val(q)\le
1$ for all quantiles $q$, we obtain the following expressions for mean
squared error in terms of the allocation function.
\begin{itemize}
\item[(i)] If the samples are generated from the first-price auction, then 
the mean squared error can be evaluated as $\MSEN[\bid] =
O(\sup_q \left\{\frac{x'(q)}{x(q)}\right\}/\sqrt{2N})$.
\item[(ii)] If the samples are generated from the all-pay auction, then 
the mean squared error can be evaluated as $\MSEN =
O(\sup_qx'(q)/\sqrt{2N})$.
\end{itemize}
\end{lemma}


To estimate the value distribution, as is evident from
equations~\eqref{eq:fp-inf} and \eqref{eq:ap-inf}, an estimator for the
derivative of the bid function, or equivalently, for the
density of the bid distribution, is needed.  Estimation of densities is standard;
however, they require assumptions on the distribution, e.g.,
continuity, and the convergence rates are strictly slower.  Our main
results do not need to explicitly estimate the value distribution and
therefore, we defer these standard methods to Section~\ref{MWN} where
our technique is compared with the straightforward approach of
estimating the value distribution explicitly.

\section{Rank-based auctions}
\label{s:iron-rank}

One of the main contributions of this paper is to introduce a
restricted class of rank-based auctions for position environments that
simultaneously have good performance and good econometric properties.
Recall that in a rank-based auction the allocation to an agent depends solely on
the rank of his bid among other agents' bids, and not on the actual
bid. For a position environment, a rank-based auction assigns agents
(potentially randomly) to positions based on their ranks.

Consider a position environment given by non-increasing weights $\wals
= (\walk[1],\ldots,\walk[n]$).  For notational convenience, define
$\walk[n+1] = 0$.  Define the cumulative position
weights $\cumwals = (\cumwalk[1],\ldots,\cumwalk[n])$ as $\cumwalk =
\sum_{j=1}^k \walk[j]$, and $\cumwalk[0]=0$. We can view the cumulative weights as defining
a piece-wise linear, monotone, concave function given by connecting
the point set $(0,\cumwalk[0]), \ldots, (n,\cumwalk[n])$.

A random assignment of agents to positions based on their ranks induces an expected weight to which agents of each rank are assigned, e.g., $\iwalk$ for the $k$th ranked agent.  These expected weights can be interpreted as a position auction environment themselves with weights $\iwals$. As for the original weights, we can define the cumulative position weights $\cumiwals$ as $\cumiwalk = \sum_{j=1}^k \walk[j]$. An important issue for optimization of
rank-based auctions is to characterize the inducible class of position weights.

\begin{lemma}[e.g., \citealp{DHH-13}]
\label{l:rank-based-feasibility}
There is a rank-based auction with induced position weights $\iwals$
for position environment with weights $\wals$ if and only if their
cumulative weights satisfy $\cumiwalk \leq \cumwalk$ for all $k$,
denoted $\cumiwals \leq \cumwals$.
\end{lemma}

Any feasible weights $\iwals$ can be constructed from a sequence of
the following two operations.
\begin{description}
\item[rank reserve] For a given rank $k$, all agents with ranks
  between $k+1$ and $n$ are rejected.  The resulting weights $\iwals$
  are equal to $\wals$ except $\iwalk[k'] = 0$ for $k' > k$.
 
\item[iron by rank] Given ranks $k' < k''$, the ironing-by-rank
  operation corresponds to, when agents are ranked, assigning the
  agents ranked in an interval $\{k',\ldots,k''\}$ uniformly at
  random to these same positions.  The ironed position weights
  $\iwals$ are equal to $\wals$ except the weights on the ironed
  interval of positions are averaged.  The cumulative ironed position
  weights $\cumiwals$ are equal to $\cumwals$ (viewed as a concave
  function) except that a straight line connects
  $(k'-1,\cumiwalk[k'-1])$ to $(k'',\cumiwalk[k''])$.  Notice that
  concavity of $\cumwals$ (as a function) and this perspective of the
  ironing procedure as replacing an interval with a line segment
  connecting the endpoints of the interval implies that $\cumwals \geq
  \cumiwals$ coordinate-wise, i.e., $\cumwalk \geq \cumiwalk$ for all
  $k$.
\end{description}

Multi-unit highest-bids-win auctions form a basis for position auctions. Consider the
marginal position weights $\margwals =
(\margwalk[1],\ldots,\margwalk[n])$ defined by $\margwalk = \walk -
\walk[k+1]$.  The allocation rule induced by the position auction
with weights $\wals$ is identical to the allocation rule induced by
the convex combination of multi-unit auctions where the $k$-unit
auction is run with probability $\margwalk$.


In this section we develop a theory for optimizing revenue over the
class of all rank-based auctions that resembles Myerson's theory for
optimal auction design.  Where Myerson's theory employs ironing by
value and value reserves, our approach analogously employs ironing by
rank and rank reserves.  We then show that the revenue of rank-based
auctions is close to optimal for position environments.  Our
econometric study of rank-based auctions is deferred to
Section~\ref{s:param-inf}.

\subsection{Optimal rank-based auctions}

In this section we describe how to optimize for expected revenue over
the class of iron by rank auctions. Recall that iron by rank auctions
are linear combinations over $k$-unit auctions. The characterization
of Bayes-Nash equilibrium, cf.\@ equation~\eqref{eq:bne-rev}, shows
that revenue is a linear function of the allocation rule. Therefore,
the revenue of a position auction can be calculated as the convex
combination of the revenues from the $k$-unit highest-bids-win
auctions.


The revenue from a $k$-unit $n$-agent highest-bids-win auction with
agent values drawn i.i.d.\@ from distribution $\dist$ can be
calculated in terms of the the agents revenue curve $\rev(\quant)$ and
the allocation rule $\knalloc(\quant)$ of the $k$-highest-bids-win
auction (for $n$ agents).  This allocation rule is precisely the
probability an agent with quantile $\quant$ has one of the highest $k$
quantiles of $n$ agents, or at most $k-1$ of the $n-1$ remaining
agents have quantiles greater than $\quant$.
$$
\knalloc(\quant) = \sum_{i=0}^{k-1} \tbinom{n-1}{i} \quant^{n-1-i}(1-\quant)^{i}.
$$ Notice that $\knalloc[0](\quant) = 0$ and $\knalloc[n](\quant) =
1$.  The {\em per-agent} revenue obtained is $\murevk =
\expect[\quant]{\marg(\quant)\,\knalloc(\quant)}$.  Notice that
$\murevk[0] = \murevk[n] = 0$.


Given the multi-unit revenues, $\murevs =
(\murevk[0],\ldots,\murevk[n])$, the problem of designing the optimal
rank-based auction is well defined: given a position environment
with weights $\wals$, find the weights $\iwals$ for an rank-based
auction with cummulative weights $\cumiwals \leq \cumwals$ maximizing
the sum $\sum_{k} (\iwalk-\iwalk[k+1])P_k$.  This optimization problem
is isomorphic to the theory of envy-free optimal pricing
developed by \citet{HY-11}.  We summarize this theory below; a
complete derivation can be found in Appendix~\ref{s:iron-opt-app}.

Define the {\em multi-unit revenue curve} as the piece-wise linear
function connecting the points $(0,\murevk[0]),\ldots,(n,\murevk[n])$.
This function may or may not be concave.  Define the {\em ironed
  multi-unit revenues} as $\imurevs =
(\imurevk[0],\ldots,\imurevk[n])$ according to the smallest concave
function that upper bounds the multi-unit revenue curve.  Define the
multi-unit marginal revenues, $\mumargs =
\mumargk[1],\ldots,\mumargk[n]$ and $\imumargs =
\imumargk[1],\ldots,\imumargk[n]$, as the left slope of the multi-unit
and ironed multi-unit revenue curves, respectively.  I.e., $\mumargk =
\murevk - \murevk[k-1]$ and $\imumargk = \imurevk - \imurevk[k-1]$.

\begin{theorem}
\label{thm:rank-based-opt}
Given a position environment with weights $\wals$, the revenue-optimal
iron-by-rank auction is defined by position weights $\iwals$ that are
equal to $\wals$, except ironed on the same intervals as $\murevs$ is
ironed to obtain $\imurevs$, and positions $k$ for which $\imumargk$ is
negative are discarded (by setting $\iwalk = 0$).
\end{theorem}

As is evident from this description of the optimal iron-by-rank
auction, the only quantities that need to be ascertained to run this
auction is the multi-unit revenue curve defined by $\murevs$.
Therefore, an econometric analysis for optimizing iron-by-rank
auctions need not estimate the entire value distribution; estimation
of the multi-unit revenues is sufficient.

\subsection{Optimal rank-based auctions with strict monotonicity}
\label{sec:strict-opt}

Position auctions, by definition, have non-increasing position weights
$\wals$.  The ironing in the iron-by-rank optimization of the
preceding section was to convert the problem of optimizing multi-unit
marginal revenue subject to non-increasing position weight, to a
simpler problem of optimizing multi-unit marginal revenue without any
constraints.  In this section, we describe the optimization of
rank-based auctions (i.e., ones for which position weights can be
shifted only downwards or discarded) subject to {\em strictly decreasing}
position weights.  This strictness is needed for insuring good
inference properties, the details of which are formalized in
Section~\ref{s:param-inf}.

As described by Lemma~\ref{l:rank-based-feasibility}, position weights
$\iwals$ are feasible as a rank-based auction in position environment
$\wals$ if the cumulative position weights satisfy $\cumwalk \geq
\cumiwalk$ for all $k$.  Suppose we would like to optimize $\iwals$
subject to a strict monotonicity constraint $\margiwalk = \iwalk -
\iwalk[k+1] \geq \epsilon$.  As non-trivial ironing by rank always
results in consecutive positions with the same weight, i.e.,
$\margiwalk = 0$ for some $k$, the optimal rank-based mechanism will
require overlapping ironed intervals.

To our knowledge, performance optimization subject to a strict
monotonicity constraint has not previously been considered in the
literature.  At a high level our approach is the following.  We start
with $\wals$ which induces the cumulative position weights $\cumwals$
which constrain the resulting position weights $\iwals$ (via its
cumulative $\cumiwals$) of any feasible rank-based auction.  We view
$\iwals$ as the combination of two position auctions. The first has
weakly monotone weights $\iyals = (\iyalk[1],\ldots,\iyalk[n])$; the
second has strictly monotone weights $((n-1)\epsilon, (n-2)\epsilon,
\ldots, \epsilon, 0)$; and the combination has weights $\iwalk =
\iyalk + (n-k)\epsilon$ for all $k$.  The revenue of the combined
position auction is the sum of the revenues of the two component position
auctions.  Since the second auction has fixed position weights, its
revenue is fixed.  Since the first position auction is weakly monotone
and the second is strictly, the combined position auction is strictly
monotone and satisfies the constraint that $\margiwalk \geq \epsilon$
for all $k$.

This construction focuses attention on optimization of $\iyals$
subject to the induced constraint imposed by $\wals$ and after the
removal of the $\epsilon$ strictly-monotone allocation rule.  I.e.,
$\iwals$ must be feasible for $\wals$.  The suggested feasibility
constraint for optimization of $\iyals$ is given by position weights
$\yals$ defined as $\yalk = \walk - (n-k) \epsilon$.  Notice that, in
this definition of $\yals$, a lesser amount is subtracted from
successive positions.  Consequently, monotonicity of $\wals$ does
not imply monotonicity of $\yals$.

To obtain $\iyals$ from $\yals$ we may need to iron for two reasons,
(a) to make $\iyals$ monotone and (b) to make the multi-unit revenue
curve monotone.  In fact, both of these ironings are good for revenue.
The ironing construction for monotonizing $\yals$ constructs the
concave hull of the cumulative position weights $\cumyals$.  This
concave hull is strictly higher than the curve given by $\cumyals$ (i.e.,
connecting $(0,\cumyalk[0]), \ldots, (n,\cumyalk[n])$).  Similarly the
ironed multi-unit revenue curve given by $\imurevs$ is the concave
hull of the multi-unit revenue curve given by $\murevs$.  The correct
order in which to apply these ironing procedures is to first (a) iron
the position weights $\yals$ to make it monotone, and second (b) iron
the multi-unit revenue curve $\murevs$ to make it concave.  This order
is important as the revenue of the position auction with weights
$\iyals$ is only given by the ironed revenue curve $\imurevs$ when the
$\margiyals = 0$ on the ironed intervals of $\imurevs$.

\begin{theorem}
\label{thm:rank-based-opt-strict}
The optimal $\epsilon$ strictly-monotone rank-based auction for position
weights $\wals$ has position weights $\iwals$ constructed by
\begin{enumerate}
\item defining $\yals$ by $\yalk = \walk - (n-k)\epsilon$ for all $k$.
\item averaging position weights of $\yals$ on intervals where $\yals$
  should be ironed to be monotone.
\item averaging the resulting position weights on intervals where
  $\murevs$ should be ironed to be concave to get $\iyals$
\item setting $\iwals$ as $\iwalk = \iyalk + (n-1) \epsilon$.
\end{enumerate}
\end{theorem}

\begin{proof} 
The proof of this theorem follows directly by the construction and its
correctness.
\end{proof}

The rank-based auction given by $\iwals$ in position environment given
by $\wals$ can be implemented by a sequence of iron-by-rank and
rank-reserve operations.  Such a sequence of operations can be found,
e.g., via an approach of \citet{AFHHM-12} or \citet{HLP-29}.

\subsection{Approximation via rank-based auctions}

In this section we show that the revenue of optimal rank-based auction
approximates the optimal revenue (over all auctions) for position
environments.  Instead of making this performance directly we will
instead identify a simple non-optimal rank-based auction that
approximates the optimal auction.  Of course the optimal rank-based
auction of Theorem~\ref{thm:rank-based-opt} has revenue at least that
of this simple rank-based auction, thus its revenue also satisfies the
same approximation bound.

Our approach is as follows.  Just as arbitrary rank-based mechanisms
can be written as convex combinations over $k$-unit highest-bids-win
auctions, the optimal auction can be written as a convex combination
over optimal $k$-unit auctions.  We begin by showing that the revenue
of optimal $k$-unit auctions can be approximated by multi-unit
highest-bids-win auctions when the agents' values are distributed
according to a regular distribution (Lemma~\ref{lem:approx-regular},
below). In the irregular case, on the other hand, rank-based auctions
cannot compete against arbitrary optimal auctions. For example, if the
agents' value distribution contains a very high value with probability
$o(1/n)$, then an optimal auction may exploit that high value by
setting a reserve price equal to that value; On the other hand, a
rank-based mechanism cannot distinguish very well between values
correspond to quantiles above $1-1/n$. We show that rank-based
mechanisms can approximate the revenue of any mechanism that does not
iron the quantile interval $[1-1/n,1]$ (but may arbitrarily optimize
over the remaining quantiles). Theorem~\ref{thm:rank-based-approx}
presents the precise statement.




\begin{lemma}
\label{lem:approx-regular}
For regular $k$-unit $n$-agent environments, there exists a $k' \leq
$k such that the highest-bid-wins auction that restricts supply to
$k'$ units (i.e., a rank reserve) obtains at least half the revenue of
the optimal auction.
\end{lemma}
\begin{proof}
This lemma follows easily from a result of \citet{BK-96} that states
that for agents with values drawn i.i.d.\@ from a regular distribution
the revenue of the $k'$-unit $n$-agent highest-bid-wins auction is at
least the revenue of the $k'$-unit $(n-k')$-agent optimal auction. To
apply this theorem to our setting, let us use $\opt{k}{n}$ to denote
the revenue of an optimal $k$-unit $n$-agent auction, and recall that
$nP_k$ is the revenue of a $k$-unit $n$-agent highest-bids-win
auction.

When $k\le n/2$, we pick $k'=k$. Then,  
$$nP_k\ge \opt{k}{n-k} \ge \frac{(n-k)}{n} \opt{k}{n}\ge \frac
12\opt{k}{n},$$ and we obtain the lemma. Here the first inequality
follows from \citeauthor{BK-96}'s theorem and the third from the
assumption that $k\le n/2$.  The second inequality follows via by
lower bounding $\opt{k}{n-k}$ by the following auction which has
revenue exactly $\frac{(n-k)}{n} \opt{k}{n}$: simulate the optimal
$k$-unit $n$-agent on the $n-k$ real agents and $k$ fake agents with
values drawn independently from the distribution.  Winners of the
simulation that are real agents contribute to revenue and the
probability that an agent is real is $(n-k)/n$.

When $k>n/2$, we pick $k'=n/2$. As before we have: 
\begin{align*}
nP_{n/2}\ge \opt{n/2}{n/2} &= \frac 12\opt{n}{n}\ge \frac 12\opt{k}{n}.\qedhere
\end{align*}
\end{proof}

\begin{lemma}
\label{lem:approx-irregular} 
For any number of agents $n$, distribution with revenue curve
$\rev(\cdot)$, and quantile $q\le 1-1/n$, there exists an integer
$k\le (1-q)n$ such that the $k$-unit $n$-agent highest-bid-wins
auction is at least a quarter of $n \rev(q)$, the revenue from posting
price $\val(q)$.
\end{lemma}
\begin{proof}
First we get a lower bound on $\murevk$ for any $k$.  For any value
$\val'$, the total expected revenue of the $k$-unit $n$-agent
highest-bid-wins auction is at least $\val'\,k$ times the probability
that $k+1$ agents have value at least $\val'$.  The median of a
binomial random variable corresponding to $n$ Bernoulli trials with
with success probability $(k+1)/n$ is $k+1$.  Thus, the probability
that this binomial is at least $k+1$ is at least $1/2$.  Combining
these observations by choosing $\val' = \val(1-(k+1)/n)$ we have,
\begin{align*}
n \, \murevk &\geq  \val(1-(k+1)/n)\, k / 2.\\
\intertext{Choosing $k = \lfloor (1-q)n\rfloor -1$, for which $\val(1-(k+1)/n) \geq \val(q)$, the bound simplifies to,}
n \, \murevk &\geq  \val(q)\, k / 2.
\end{align*}
The ratio of $P_k$ and $\rev(q) = (1-q)\,\val(q)$ is therefore at least $k/2(1-q)n> k/2(k+2)$, which for $q\le 1-3/n$ (or, $k\ge 2$) is at least $1/4$.

For $q\in (1-3/n,1-1/n]$, we pick $k=1$. Then, $\murevk[1]$ is at least $1/n$
  times $\val(q)$ times the probability that at least two agents have a
  value greater than or equal to $\val(q)$. We can verify for $n\ge 2$ that
$$\murevk[1] \ge \frac {\val(q)}n  \left( 1-q^n-n(1-q)q^{n-1} \right)\ge \frac 14 (1-q)\,\val(q).$$
\end{proof}

\begin{theorem}
\label{thm:rank-based-approx}
For regular agents and position environments, the optimal rank-based
auction obtains at least half the revenue of the optimal auction.  For
(possibly irregular agents) and position environments, optimal rank-based
auction obtains at least a quarter of the revenue of the optimal auction that does not iron or set a reserve price on the quantile interval $[1-1/n,1]$.
\end{theorem}
\begin{proof}
In the regular setting, the theorem follows from
Lemma~\ref{lem:approx-regular} by noting that the optimal auction
(that irons by value and uses a value reserve) in a position
environment is a convex combination of optimal $k$-unit auctions:
since the revenue of each of the latter can be approximated by that of
a $k'$-unit highest-bids-win auction with $k'\le k$, the revenue of
the convex combination can be approximated by that of the same convex
combination over $k'$-unit highest-bids-win auctions; the resulting
convex combination over $k'$-unit auctions satisfies the same position
constraint as the optimal auction.

In the irregular setting, once again, any auction in a position
environment is a convex combination of optimal $k$-unit auctions. The
expected revenue of any $k$-unit auction is bounded from above by the
expected revenue of the optimal auction that sells at most $k$ items
in expectation. The per-agent revenue of such an auction is bounded by
$\irev(1-k/n)$, the revenue of the optimal allocation rule with ex
ante probability of sale $k/n$. Here $\irev(\cdot)$ is the ironed
revenue curve (that does not iron on quantiles in
$[1-1/n,1]$). $\irev(1-k/n)$ is the convex combination of at most two
points on the revenue curve $\rev(a)$ and $\rev(b)$, $a\le 1-k/n\le b
< 1-1/n$.  Now, we can use Lemma~\ref{lem:approx-irregular} to obtain
an integer $k_a < n(1-a)$ such that $P_{k_a}$ is at least a quarter of
$\rev(a)$, likewise $k_b$ for $b$. Taking the appropriate convex
combination of these multi-unit auctions gives us a $4$-approximation
to the optimal auction $k$-unit auction (that does not iron over the
quantile interval $[1-1/n,1]$).  Finally, the convex combination of
the multi-unit auctions with $k_a$ and $k_b$ corresponds to a position
auction with that is feasible for a $k$ unit auction (with respect to
serving the top $k$ positions with probability one, service
probability is only shifted to lower positions).
\end{proof}

\section{Inference in rank-based auctions}
\label{s:param-inf}

Recall that the performance of any rank-based auction is governed by
the multi-unit revenues $\murevk[1],\ldots,\murevk[n]$ with $\murevk$
equal to the per-agent revenue of the highest-$k$-agents-win auction. In order
to optimize over the class of rank based auctions, then, we need to
estimate the $n$ quantities $\murevk$. We now describe how to estimate
these quantities from the observed bids, and how the error in the
estimation of the bid distribution translates into errors in the
estimated multi-unit revenues.

Let $\alloc$ denote the allocation rule of the auction that we run,
and let $\bid$ denote the bid distribution in BNE of this auction.
Recall that $\knalloc(\cdot)$ denotes the allocation rule of the
highest-$k$-agents-win auction. In the following, we use $\kalloc$ as
a short-form for $\knalloc$. Then, the per-agent revenue of this auction is
given by:
$$
\murevk = \expect[\quant]{\kalloc'(\quant)\rev(\quant)} =
\expect[\quant]{\kalloc'(\quant)\val(\quant) (1-\quant)}
$$

We will now perform our analysis for the all-pay and first-price
auction formats separately, using the respective bid-to-value
conversion equations from Section~\ref{s:prelim}.

\subsection{Inference for an all-pay auction}

Recall that for an all-pay auction format, we can convert the bid
distribution into the value distribution as follows:
$\val(\quant) = \bid'(\quant)/\alloc'(\quant)$.
Substituting this into the expression for $P_k$ above we get
$$
\murevk =
\expect[\quant]{\kalloc'(\quant)(1-\quant)\frac{\bid'(\quant)}{\alloc'(\quant)}}
= \expect[\quant]{Z_k(\quant)\bid'(\quant)}
$$
where $Z_k(\quant)=(1-\quant)\frac{\kalloc'(\quant)}{\alloc'(\quant)}$.

Writing the expectation as an integral and integrating by parts we
obtain the following lemma. Here we note that $\bid(0)=0$ and $Z_k(1)=0$.

\begin{lemma}\label{weights lemma}
  The per-agent revenue of the highest-$k$-agents-win auction can be written as
  a linear combination of the bids in an all-pay auction: 
$$
\murevk = \expect[\quant]{-Z'_k(\quant)\bid(\quant)}
$$
where $Z_k(\quant)
=(1-\quant)\frac{\kalloc'(\quant)}{\alloc'(\quant)}$ depends on the
allocation rule of the mechanism and is known precisely.
\end{lemma}

This formulation allows us to express the error in estimation of $P_k$
in terms of the error in estimating the bid distribution. In
particular, let $\emurevk$ denote the estimate of $P_k$ obtained by
plugging the bid estimator $\widehat{\bid}(\cdot)$ in the formula
given by Lemma~\ref{weights lemma}.
%
Then we can write the error in $\murevk$ as:
\begin{eqnarray*}
|\emurevk-\murevk| & = &
\expect[q]{\left|-Z'_k(\quant)(\widehat{\bid}(\quant)-\bid(\quant))\right|}
\le
\expect[q]{|Z'_k(\quant)|}\,\sup_q |\widehat{\bid}(\quant)-\bid(\quant)|
\end{eqnarray*}
Lemma~\ref{error bid function} then gives the following bound on the
mean squared error for $\murevk$:

\begin{eqnarray*}
\MSEN[\murevk] & \le & \frac{\sup_{\quant}b'(\quant)}{\sqrt{2\samples}} \expect[\quant]{|Z'_k(\quant)|}  
\end{eqnarray*}

We now proceed to bound $\expect[\quant]{|Z'_k(\quant)|}$.  To this
end we first note that if $\alloc$ is a convex combination over the
allocation rules of the multi-unit highest-bids-win auctions, then
$Z_k$ has a single local maximum (see the appendix for a proof).

\begin{lemma}
\label{lem:Z-bound-1}
Let $\kalloc$ denote the allocation function of the
$k$-highest-bids-win auction and $\alloc$ be any convex combination
over the allocation functions of the multi-unit auctions. Then the
function $Z_k(\quant)=(1-\quant)
\frac{\kalloc'(\quant)}{\alloc'(\quant)}$ achieves a single local
maximum for $\quant\in [0,1]$.
\end{lemma}

Let $Z_k^* = \sup_{\quant} Z_k(\quant)$. Then, we can bound
$\expect[\quant]{|Z'_k(\quant)|}$ by $2Z_k^*-Z_k(1)-Z_k(0)\le
2Z_k^*$. We get the following theorem:

\begin{theorem}\label{th: all pay}
  Let $\kalloc$ denote the allocation function of the
  $k$-highest-bids-win auction and $\alloc$ be any convex combination
  over the allocation functions of the multi-unit auctions. Then for
  all $k$, the mean squared error in estimating $P_k$ from $\samples$
  samples from the bid distribution for an all-pay auction with
  allocation rule $\alloc$ is:
$$\MSEN[\murevk]\le \sqrt{\frac{2}{\samples}}\,\, \sup_{\quant}\{\alloc'(\quant)\}\,\,\sup_{\quant}
\left\{ \frac{(1-\quant)\kalloc'(\quant)}{\alloc'(\quant)} \right\}
$$.
\end{theorem}

\subsection{Inference from a first-price auction}
Recall that in a first-price auction, we can obtain the value
distribution from the bid distribution as follows: $\val(\quant) =
\bid(\quant) +
\alloc(\quant)\bid'(\quant)/\alloc'(\quant)$. Substituting this into
the expression for $\murevk$ we get:
$$
\murevk = \expect[\quant]{(1-\quant)\kalloc'(\quant) \bid(\quant) +
  \frac{(1-\quant)\kalloc'(\quant)
    \alloc(\quant)\bid'(\quant)}{\alloc'(\quant)}} =
\expect[\quant]{(1-\quant)\kalloc'(\quant) \bid(\quant) + Z_k(\quant) \alloc(\quant)\bid'(\quant)}
$$
where, as before, $Z_k(\quant) =
\frac{(1-\quant)\kalloc'(\quant)}{\alloc'(\quant)}$.

Integrating the second expression by parts, we get
\begin{eqnarray*}
\int_0^1 Z_k(\quant) \alloc(\quant)\bid'(\quant)\,d\quant & = &
Z_k(\quant) \alloc(\quant)\bid(\quant)|_0^1 - \int_0^1 (Z'_k(\quant) \alloc(\quant)+Z_k(\quant) \alloc'(\quant))\bid(\quant) \,d\quant\\
& = & - \int_0^1 Z'_k(\quant) \alloc(\quant)\bid(\quant)\,d\quant -  \int_0^1 (1-\quant) \kalloc'(\quant)\bid(\quant)\,d\quant 
\end{eqnarray*}

When we put this back in the expression for $P_k$ two of the terms
cancel, and we get the following lemma.
\begin{lemma}
  The per-agent revenue of the highest-$k$-agents-win auction can be written as
  a linear combination of the bids in a first-pay auction: 
$$
\murevk = \expect[\quant]{-\alloc(\quant)Z'_k(\quant)\bid(\quant)}
$$
where $Z_k(\quant)
=(1-\quant)\frac{\kalloc'(\quant)}{\alloc'(\quant)}$ and
$\alloc(\quant)$ are known precisely.
\end{lemma}

To bound the error in estimating $P_k$, once again we need to bound
the integral $\int_0^1 \alloc(\quant) |Z'_k(\quant)|\,d\quant$. Recall
that $Z_k$ has a single local maximum for $\quant\in [0,1]$ when
$\alloc$ is a convex combination over the multi-unit auctions. This
implies the following lemma (see the appendix for a proof).

\begin{lemma}
\label{lem:Z-bound-2}
  $\int_0^1 \alloc(\quant) |Z'_k(\quant)|\,d\quant\le
  2\alloc(q_k^*)Z_k(q_k^*)+1$ where $q_k^* = \argmax_{\quant} Z_k(\quant)$.
\end{lemma}

Using this lemma and applying Lemma \ref{error bid function} from Section~\ref{s:prelim} we
get the following theorem.

\begin{theorem}\label{th: first price}
Let $\kalloc$ denote the allocation function of the
  $k$-highest-bids-win auction and $\alloc$ be any convex combination
  over the allocation functions of the multi-unit auctions. Then for
  all $k$, the mean squared error in estimating $P_k$ from $\samples$
  samples from the bid distribution for a first-price auction with
  allocation rule $\alloc$ is:
$$\MSEN[\murevk]\le
\sqrt{\frac{2}{\samples}}\sup_{\quant}\left\{ \frac{\alloc'(\quant)}{\alloc(\quant)}\right\} \sup_{\quant}
\left\{
  \frac{(1-\quant)\alloc(\quant)\kalloc'(\quant)}{\alloc'(\quant)}
\right\}
$$.
\end{theorem}

\subsection{Revenue versus inference tradeoff for rank-based auctions}

We now consider optimizing for expected revenue over the class of rank
based auctions subject to good inferability of the parameters
$P_k$. Recall that the revenue of a rank based auction with
position weights $\wals$ and marginal weights
$\margwals=\walk-\walk[k+1]$ is given by $\sum_{k} \margwalk P_k$. On
the one hand, estimating the $P_k$'s well is important to be able to
optimize $\wals$ -- we should place the most marginal weight on
positions with high $P_k$'s. On the other hand, the weights $\wals$
determine the allocation rule $\alloc$ as a weighted sum of the
$k$-unit allocation rules $\kalloc$, which in turn via
Theorems~\ref{th: all pay} and \ref{th: first price} determine the
error in the $P_k$'s -- we should ensure that all positions get some
minimum marginal weight. This is the problem of finding the optimal
$\epsilon$ strictly monotone rank-based auction that we discussed in
Section~\ref{sec:strict-opt} (see
Theorem~\ref{thm:rank-based-opt-strict}).

We now claim that the auction returned by
Theorem~\ref{thm:rank-based-opt-strict} obtains revenue close to the
optimal rank-based auction. In particular, one way of obtaining an
$\epsilon$ strictly monotone auction given the estimates
$\widehat{P_k}$ is to run the optimal auction with probability
$1-\eps$ and with probability $\eps$ run the auction that assigns
equal marginal weight to every position. In particular, for every $k$,
$\margwalk\ge \eps/n$. For this auction, recall that $\alloc = \sum_k
\margwalk \kalloc$, and $\alloc' = \sum_k \margwalk
\kalloc'$. Therefore, for any quantile $\quant$,
$$ \kalloc'(\quant)/\alloc'(q)\le 1/\margwalk \le n/\eps.$$ We obtain
the following theorem.

\begin{theorem}
\label{thm:param-rate}
  For every $\eps>0$, there exists a rank based auction on $n$
  agents that obtains a $1-\eps$ approximation to the optimal 
  rank based revenue. Furthermore, from $\samples$ samples of the bid
  distribution, we can estimate parameters $P_k$ for $k\in [n]$ with
  error bounds as below:
\begin{align*}
\text{For the first price format: } & 
|\emurevk-\murevk|\le
\sqrt{\frac{2}{\samples}}\,\,\frac{n}{\eps}\,\,\sup_{\quant}\left\{ \frac{\alloc'(\quant)}{\alloc(\quant)}\right\} \\
\text{For the all pay format: } & 
|\emurevk-\murevk|\le
\sqrt{\frac{2}{\samples}}\,\,\frac{n}{\eps}\,\,\sup_{\quant}\left\{ \alloc'(\quant)\right\} 
\end{align*}
\end{theorem}

We remark that while the theorem above gives the same upper bound on
the error in estimation for $P_k$s for both the first price and all
pay auction formats, comparing the bounds in 
Theorems~\ref{th: all pay} and \ref{th: first price} 
shows that the first price format is
better at inference than the all pay format.  Note that the error bound
can be made arbirarily small by picking a large enough sample size
$\samples$.

\section{Inferring the revenue curve}\label{MWN}

\def \std {\text{std}}
\def \eps {\epsilon}
\def \cone {C_1}
\def \ctwo {C_2}
\def \cthree {C_3}

In the previous section we considered the problem of inferring the
parameters of the position auction from samples obtained from a first
price and an all-pay auction.  We now consider the problem of infering
the entire revenue curve $\rev(\cdot)$ from bid samples by first
infering the value distribution. This is relevant, for instance, if we
want to estimate the revenue of an arbitrary mechanism, and if we want
to optimize for revenue over the class of all mechanisms. We find that
the inference problem becomes harder and a tight bound on revenue
requires polynomially more samples as compared to the previous
setting. 

\subsection{Propagation of errors in inference of value distribution}

In order to infer the value distribution from the bid distribution, as
given in Section~\ref{s:prelim} by equations \eqref{eq:fp-inf} and
\eqref{eq:ap-inf}, we need to estimate the derivative of the bid
function $b(\cdot)$. Let $G(\cdot)$ denote the c.d.f. for bids, that is, $G(z)=b^{-1}(z)$
is the probability that a random bid is no more than $z$. Let
$g(z) = \frac{d}{dz} b^{-1}(z)$ denote the corresponding density
function. Note that $g(b(q)) = 1/b'(q)$. We will therefore focus on
estimating $g(\cdot)$.




The density $g(\cdot)$ cannot be estimated directly from the empirical
bid distribution of equation \eqref{bid function} because the
derivative of that distribution is undefined. Nonetheless, a number of
standard estimators are available to estimate $g(\cdot)$. Using such
an estimator, $\hat{g}(\cdot)$, we obtain an estimator for the
derivative as follows:
\begin{equation}\label{derivative:estimator}
\ebid'(\quant)=1/\ghat(\ebid(\quant)).
\end{equation}

In Appendix~\ref{s:b-prime-estimator} we formally state the requirements for an estimator
of $b'(\cdot)$.
We assume that we know the rate of
convergence for this estimator, i.e. the sequence $\rate$ with
$\rate \rightarrow \infty$ as we obtain more samples, such that:
$$
E\left[\sup\limits_b|\ghat(b)-g(b))|^2\right]^{1/2}=O(1/\rate).
$$
 For instance, if one uses the histogram-based
estimator for the density of bids, then
$$
\frac{1}{\hat{b}'(q)}=\frac{1}{\samples h}\sum\limits^{\samples}_{i=1}
{\bf 1}\left\{|\hat{b}(q)-b_{i}|\leq h\right\},
$$
where $h$ is the bandwidth, which is selected such that 
$h \,\samples/\log(\samples) \rightarrow \infty$ as we get more samples. 
In this case given that the class
of indicators ${\bf 1}\left\{|b-t|\leq h\right\}$ when
$t \in [b-\epsilon,b+\epsilon]$
(which depends on $\samples$) has a metric entropy of order $O(\epsilon)$,
then the estimator for the derivative of the quantile function
of the bid distribution converges at rate $\rate=\sqrt{\samples\,h}$.

Functional objects, such as distribution densities,
can be estimated using many estimators. If we
restrict ourselves to feasible estimators (that are
completely data-driven) and avoid oracle estimators,
we can talk about an feasible estimator that achieves
the fastest convergence rate. Such
a convergence rate is called the {\it optimal convergence rate}. 
\cite{stone} established the optimal conevergence rate for
estimation of one-dimensional density, which we formulate
here without proof.

\begin{lemma}\label{optimal rate}
Suppose that the density of bids $g(\cdot)$ has $k$ derivatives.
Then the optimal convergence rate for the estimator 
for the density $\ghat(\cdot)$ is $\rate=\samples^{k/(1+2k)}$.
\end{lemma}

This theorem implies that there is a lower bound on the convergence
rate equal to ${\samples}^{1/3}$ for estimation of distribution densities
that have one derivative. At the same time, for functions that are
very smooth, the optimal convergence rate can approach the maximum
rate of ${\samples}^{1/2}$. 

We now establish the mean-squared error of the estimator for the
derivative of the bid distribution.
\begin{theorem}\label{th:derivative}
The mean-squared error for the estimator (\ref{derivative:estimator}) for the derivative of the 
bid function at quantile $\quant$ can be represented as
$$
\MSEN[b'(q)]=E\left[(\ebid'(\quant)-\bid'(\quant))^2\right]^{1/2}
=O\left(
\frac{\bid'(\quant)^2}{\rate}+\frac{|\bid^{\prime\prime}(\quant)|}{\sqrt{2 \samples}}
\right)
$$
\end{theorem}

%


The value function can now be estimated using Equation
\eqref{eq:fp-inf} or \eqref{eq:ap-inf}, as applicable.

\subsection{Inference from an all pay auction}

Recall that for rank based mechanisms, we know the allocation function
$\alloc(q)$ and its derivative $x'(q)$ precisely. 
For all pay auctions Equation~\eqref{eq:ap-inf} allows us to relate
the value distribution to the bid distribution: $v(q)=b'(q)/x'(q)$.

Let $\MSEN[b'(q)]=E\[(\hat{b}'(q)-b'(q))^2\]^{1/2}$ and 
$\MSEN[v(q)]=E\[|\vhat(q)-v(q)|^2\]^{1/2}$ denote the mean squared
error in $b'$ and $v$ respectively.
%
%
%

We can express the mean squared error in $\rev$,
$\MSEN[R(q)]=E\[(\hat{R}(q)-R(q))^2\]^{1/2}$, in terms of the error in
$b'$ as follows.
We write the revenue at quantile $q$ as $R(q) = (1-q)v(q)$ meaning that 
we can estimate the revenue as $\hat{R}(q)=(1-q)\vhat(q)$ by replacing
the true value with its estimated counterpart. 
then
$$
\frac{\MSEN[R(q)]}{R(q)}=\frac{(1-q)\MSEN[v(q)]}{(1-q)v(q)} 
=\frac{\MSEN[b'(q)]/x'(q)}{b'(q)/x'(q)}= \frac{\MSEN[b'(q)]}{b'(q)}.
$$ 

Our goal is to bound this relative error by a quantile-dependent
quantity, $\eps(q)$.
Formally, we require that
$$
\frac{\MSEN[R(q)]}{R(q)} \leq \eps(q),
$$
assuming that $R(q)>0$. Our discussion above demonstrates that this
error in turn can be expressed in the relative error of estimation of
the derivative of the bid function $\MSEN[b'(q)]/b'(q)$. We can now use
Theorem~\ref{th:derivative} to obtain an error bound for the revenue
curve\footnote{Proofs for this section can be found in the appendix.}
and derive conditions on the allocation rule $x$ that bound the
relative error in $b'(q)$ by $\eps(q)$ at all quantiles $q$:




\begin{theorem}\label{theorem1}
 Suppose that the allocation rule $x(\cdot)$ of an all-pay auction along with its second and first
  derivatives satisfies for all quantiles $q$:
$$
\frac{\Omega(1)|\alloc^{\prime\prime}(\quant)|}{\eps(q) \sqrt{\samples}} \leq x'(q) \leq
O(1)\frac{\eps(q)\rate}{v(q)}, \; \text{and, }
 \samples\geq \frac12\left(\frac{v'(\quant)}{v(\quant)\epsilon(\quant)}\right)^2.
$$
where $\rate$ is the convergence rate for the estimator of the bid
density $g(\cdot)$.
Then, the relative error in estimating the revenue curve from
$\samples$ samples of the bid distribution is bounded by function
$\eps(q)$.
\end{theorem}

Let us consider the bounds on $x'$ closely. The lower bound of this
expression is determined by the curvature of the allocation function
and the number of samples. It requires the allocation rule to
``separate" bids within the range
$|\alloc^{\prime\prime}(\quant)|/\sqrt{2 \samples}$.  The upper bound
in this expression is driven by the sampling noise in the inference of
the density of bids: if the allocation rule ``jumps'' at a certain
quantile, the density of bids at that quantile is low and the relative
error in bid density due to sampling becomes large. Note also that the
lower bound on the number of samples required is determined by the
slope of the value function, with more samples required for more
concentrated distributions.

From the upper bound on $x'(q)$ we note that the error in estimating
the revenue curve is at best $v(q)x'(q)/\rate$, that is, we estimate
the revenue at a rate of at most $\rate$.\footnote{This
  simplification ignores whether an allocation rule achieving this
  rate exists.} Recall from Lemma~\ref{optimal rate} that $\rate$ can be as small
as $N^{1/3}$ if $x$ is not sufficiently smooth. This rate is much
slower than the parametric convergence rate $\sqrt{\samples}$ derived
in Theorem~\ref{thm:param-rate}
and means that inference for the multi-unit revenues can be performed with
much fewer samples than inference for general mechanisms.


\subsection{First price auctions}
The analysis for the first-price auctions follows closely
our analysis for all-pay auctions.
Recall that the value function can be obtained
from the bid function and its derivative as 
$$
v(q) = b(q)+\frac{x(q)}{x'(q)}b'(q).
$$
The value function is estimated by replacing the 
bid function and its derivative with their estimated
counterparts.
We further notice that the relative impact on the expected revenue
can be bounded in the same fashion as for the all-pay auctions, meaning
that $$
\frac{\MSEN[R(q)]}{R(q)}=\frac{\MSEN[v(q)]}{v(q)}.
$$
This allows us to write an analog of Theorem \ref{theorem1} for
the first-price auctions.

\begin{theorem}\label{theorem2}
Suppose that the allocation rule $x(\cdot)$ of a first-price auction 
along with its second and first derivatives satisfies for all
quantiles $q$:
$$
\frac{\Omega(1)|\alloc^{\prime\prime}(\quant)|}{\sqrt{\samples}\eps(\quant)}
\leq x'(\quant) \leq O(1) \alloc(\quant) \rate \eps(\quant),\;\; \samples \geq \frac12
\left(
\frac{v'(\quant)}{\epsilon(\quant)}
\right)^2.
$$
where $\rate$ is the convergence rate for the estimator of the bid
density $g(\cdot)$.
Then the relative error in estimating the revenue curve from
$\samples$ samples of the bid distribution is bounded by function $\eps(q)$.
\end{theorem}

Once again from the upper bound on $x'(q)$, we get that the error
$\eps(q)$ is at least $x'(q)/(\rate x(q))$, that is, the revenue
curve can be estimated at best at a rate of $\rate$, which by
Lemma~\ref{optimal rate} can be as small as $N^{1/3}$.




\section{Discussion and Conclusions}

We conclude with some observations and discussion.
\begin{itemize}
\item Good inference requires careful design of the mechanism. Perfect
  inference and perfect optimality cannot be achieved together.

\item We cannot achieve good accuracy in infering the revenue of an
  arbitrary mechanism, or in infering the entire revenue curve. In
  contrast, the multi-unit revenues $P_k$ are special functions that
  depend linearly on the bid distribution (and not, for example, on
  bid density). This property enables them to be learned accurately.

\item Rank based mechanisms achieve a good tradeoff between revenue
  optimality and quality of inference in position environments: (1)
  They are close to optimal regardless of the value distribution;
  (2) Optimizing over this class for revenue requires estimating only
  $n$ parameters $P_k$ that, by our observation above, are ``easy'' to
  estimate accurately; (3) Rank based mechanisms satisfy the
  necessary conditions on the slope of the allocation function that
  enable good inference.

\end{itemize}

\bibliographystyle{acmsmall}
\bibliography{uniq,agt,references}

\appendix

\section{Finding the optimal iron by rank auction}
\label{s:iron-opt-app}

Recall that iron by rank auctions are weighted sums of multi-unit
auctions. Therefore, their revenue can be expressed as a weighted sum
over the revenues $P_k$ of $k$-unit auctions. We consider a position environment given by non-increasing weights $\wals
 = (\walk[1],\ldots,\walk[n]$), with $\walk[0] = 0$, $\walk[1]=1$, and $\walk[n+1] = 0$.  Define the cumulative position
weights $\cumwals = (\cumwalk[1],\ldots,\cumwalk[n])$ as $\cumwalk =
\sum_{j \leq i} \walk[j]$.

Define the {\em multi-unit revenue
  curve} as the piece-wise constant function connecting the points
$(0,\murevk[0],\ldots,(n,\murevk[n])$.  This function may or may not
be concave.  Define the {\em ironed multi-unit revenue curve} as
$\imurevs = (\imurevk[1],\ldots,\imurevk[n])$ the smallest concave
function that upper bounds the multi-unit revenue curve.  Define the
multi-unit marginal revenues as $\mumargs =
\mumargk[1],\ldots,\mumargk[n]$ and $\imumargs =
\imumargk[1],\ldots,\imumargk[n]$ as the left slope of the multi-unit
and ironed multi-unit revenue curves, respectively.  I.e., $\mumargk = \murevk - \murevk[k-1]$ and $\imumargk = \imurevk - \imurevk[k-1]$.

We now see how the revenue of any position auction can be expressed in
terms of the multi-unit revenue curves and marginal revenues.
\begin{align*}
\expect{\text{revenue}} &= \sum_{k=0}^n \murevk\,\margwalk 
                         = \sum_{k=0}^n \mumargk\,\walk\\
                        &\leq \sum_{k=0}^n \imurevk\,\margwalk 
                         = \sum_{k=0}^n \imumargk\,\walk.
\end{align*}
The first equality follows from viewing the position auction with
weights $\wals$ as a convex combination of multi-unit auctions (where
its revenue is the convex combination of the multi-unit auction
revenues).  The second and final inequality follow from rearranging
the sum (an equivalent manipulation to integration by parts).  The
inequality follows from the fact that $\imurevs$ is defined as the
smallest concave function that upper bounds $\murevs$ and, therefore,
satisfies $\imurevk \geq \murevk$ for all $k$.  Of course the
inequality is an equality if and only if $\margwalk = 0$ for every
$k$ such that $\imumargk > \mumargk$.

We now characterize the optimal ironing-by-rank position auction.
Given a position auction weights $\wals$ we would like the
ironing-by-rank which produces $\iwals$ (with cumulative weights
satisfying $\cumwals \geq \cumiwals$) with optimal revenue.  By the
above discussion, revenue is accounted for by marginal revenues, and
upper bounded by ironed marginal revenues.  If we optimize for ironed
marginal revenues and the condition for equality holds then this is
the optimal revenue.  Notice that ironed revenues are concave in $k$,
so ironed marginal revenues are monotone (weakly) decreasing in $k$.
The position weights are also monotone (weakly) decreasing.  The
assignment between ranks and positions that optimizes ironed marginal
revenue is greedy with positions corresponding to ranks with negative
ironed marginal revenue discarded.  Tentatively assign the $k$th rank
agent to slot $k$ (discarding agents that correspond to discarded
positions).  This assignment indeed maximizes ironed marginal revenue
for the given position weights but may not satisfy the condition for
equality of revenue with ironed marginal revenue.  To meet this
condition with equality we can randomly permute (a.k.a., iron by rank)
the positions that corresponds to intervals where the revenue curve is
ironed.  This does not change the surplus of ironed marginal revenue
as the ironed marginal revenues on this interval are the same, and the
resulting position weights $\iwals$ satisfy the condition for equality
of revenue and ironed marginal revenue.

\section{Requirements for the estimator of the derivative of the bid function}
\label{s:b-prime-estimator}

In general,
we can express $$
p'(b)=\Phi_n(b;G,g),\;\;\mbox{and}\;\;
x'(b)=\Psi_n(b;G,g),
$$
where $G(\cdot)$ is the cdf of the distribution of bids
and $g(\cdot)$ is the pdf. For instance, when the mechanism $\mathcal M$
is the first-price auction, then $\Phi_n(b;G,g)=n\,G^{n-1}(b)\,g(b)$
and $\Psi_n(b;G,g)=G^{n}(b)+n\,b\,G^{n-1}(b)g(b)$. For the all-pay
auction $\Phi_n(b;G,g)=n\,G^{n-1}(b)\,g(b)$ and $\Psi_n(b;G,g)=1$.

\begin{assumption}\label{assume1}
Suppose that 
\begin{itemize}
\item[(i)] Suppose that the density of bids $g$ is bounded by a universal constant $\bar g$. 
There exists an estimator $\widehat{g}$ for which 
$r\,(\widehat{g}-g)(\cdot)$ converges to a tight stochastic process
with convergence rate $r$ such that $r \rightarrow \infty$  $r/\sqrt{nT} \rightarrow 0$
\item[(ii)] $\Phi_n(b;G,g)$ and $\Psi_n(b;G,g)$ are smooth 
functionals of $G$ and $g$ for each $b$ such that for any two pairs
$(G_1,g_1)$ and $(G_2,g_2)$ with $\|G_1-G_2\|\leq \varepsilon_1$
and $\|g_1-g_2\|\leq \varepsilon_2$:
$$
\sup\limits_{b\in [0,\bar v]}\left|\Phi_n(b;G_1,g_1)-\Phi_n(b;G_2,g_2)\right|
\leq J_n^{\Phi,1}\varepsilon_1+J_n^{\Phi,2}\varepsilon_2
$$
and
$$
\sup\limits_{b\in [0,\bar v]}\left|\Psi_n(b;G_1,g_1)-\Psi_n(b;G_2,g_2)\right|
\leq J_n^{\Psi,1}\varepsilon_1+J_n^{\Psi,2}\varepsilon_2
$$
\end{itemize}
\end{assumption}

We imposed this high-level assumption to facilitate a wide range of
estimators that can be used to estimate the distribution and the density
of the distribution of bids.

\section{Proofs}
\begin{proofof}{Lemma~\ref{error bid function}}
Consider estimation of the bid function using the sorted 
bids $b^{(1)} \geq b^{(2)} \geq  \ldots \geq b^{(N)}$.
Then the bid function is estimated as
$$\widehat{b}(q)=b^{([qN])},$$
where $[\cdot]$ is the floor integer. We can equivalently
express this function as a solution of the following equation
$$\frac{1}{N}\sum^N_{i=1}{\bf 1}\{b_i \leq \widehat{b}(q)\}=q+o_p(1/\sqrt{N}),$$
where $o_p(1/\sqrt{N})$ corresponds to the error that arizes because the empirical
cdf is a step function and identifies the true cdf in the steps of size $1/N$.
Let $G(\cdot)$ be the cdf of bids and $\widehat{G}(\cdot)$ be the 
empirical cdf. Then the equation above can be rewritten as
$\widehat{G}(\widehat{b}(q))-G(b(q))=o_p(1/\sqrt{N}).$
Now we decompose this expression as
$$
\widehat{G}(\widehat{b}(q))-G(b(q))=\widehat{G}(\widehat{b}(q))-G(\widehat{b}(q))+{G}(\widehat{b}(q))-G(b(q)).
$$
By the Donsker theorem $\sqrt{N}(\widehat{G}(t)-G(t))$ converges to 
a tight mean zero stochastic process ${\mathbb G}(t)$ over $t$ with covariance function
such that $H(t,t)=G(t)(1-G(t))$.
Note that 
$$\sup\limits_{t}|H(t,t)| \leq \frac14.$$
This means that 
$$
E\left[\left( \sqrt{N}\sup\limits_{t}(\widehat{G}(t)-G(t)) \right)^2\right] \leq \frac14.
$$
Next, consider the following expansion:
$$
{G}(\widehat{b}(q))-G(b(q))=g(b(q))(\widehat{b}(q)-b(q))+o(|\widehat{b}(q)-b(q)|^2).$$
Combining this result together with the decomposition above and recalling
that $b'(q)=1/g(b(q))$, we write
$$
\sqrt{N}(\widehat{b}(q)-b(q))=-b'(q)\sqrt{N}(\widehat{G}(\widehat{b}(q))-G(\widehat{b}(q)))+o_p(1).
$$
Then we write
$$
E\left[\left( \sqrt{N}\sup\limits_{q}(\widehat{b}(q)-b(q))\right)^2\right] \leq \sup_qb'(q)E\left[\left( \sqrt{N}\sup\limits_{t}(\widehat{G}(t)-G(t)) \right)^2\right]^{1/2}.
$$
This means that
$$\MSEN[b] \leq \frac{\sup_qb'(q)}{2}\,\frac{1}{\sqrt{N}}.$$
Then, recalling that $v(q) \leq 1$ and $v(q)-b(q) \leq 1$, we
can replace the upper bound on $b'(q)$ by $x'(q)$ for an all-pay
auction and by $x'(q)/x(q)$ for the first-price auction.
\end{proofof}




\begin{proofof}{Lemma~\ref{lem:Z-bound-1}}
  Consider the function $A(q) = 1/Z_k(q) =
  x'(q)/(1-q)x'_k(q)$. $x'(q)$ is a weighted sum over $x'_j(q)$ for
  $j\in\{1,\cdots,n-1\}$. So, $A(q)$ is a weighted sum over terms
  $x'_j(q)/(1-q)x'_k(q)$. Let us look at these terms closely.
$$
\frac{x'_j(q)}{(1-q)x'_k(q)} = \alpha_{k,j} q^{k-j}(1-q)^{j-k-1}
$$
where $\alpha_{k,j}$ is independent of $q$. The functions
$q^{k-j}(1-q)^{j-k-1}$ are convex. This implies that $A(q)$ which is a
weighted sum of convex functions is also convex. Consequently, it has
a unique minimum. Therefore, $Z_k(q) = 1/A(q)$ has a unique maximum.
\end{proofof}

\begin{proofof}{Lemma~\ref{lem:Z-bound-2}}
Recall that $Z_i(q)$ has a single local maximum at quantile $q_i^*$. So we get 
$$
\int^1_0x(q)|Z'_i(q)|\,dq = \int^{q_i^*}_0x(q)Z'_i(q)\,dq - \int^{1}_{q_i^*}x(q)Z'_i(q)\,dq 
$$

Integrating by parts, 
$$
\int x(q)Z'_i(q)\,dq = x(q)Z_i(q) - \int
x'(q)Z_i(q)\, dq = x(q)Z_i(q) - \int (1-q) x'_i(q)\, dq
$$

Therefore, 
$$
\int^1_0x(q)|Z'_i(q)|\,dq = 2x(q_i^*)Z_i(q_i^*) - \int^{q_i^*}_0 (1-q)
x'_i(q)\, dq + \int^1_{q_i^*} (1-q) x'_i(q)\, dq
< 2x(q_i^*)Z_i(q_i^*) + 1
$$
\end{proofof}

\begin{proofof}{Theorem~\ref{th:derivative}}
Consider the difference 
$$
\bhat'(\quant)-b'(\quant)=\frac{1}{\ghat(\bhat(\quant))}-\frac{1}{g(\bhat(\quant))}+
\frac{1}{g(\bhat(\quant))}-\frac{1}{g(b(\quant))}.
$$
Using the Taylor expansion, with probability approaching 1 we can bound
$$
|\bhat'(\quant)-b'(\quant)| \leq \frac{1}{g(b(\quant))^2}|\ghat(\bhat(\quant))-g(\bhat(\quant))|
+|\frac{g'(b(\quant))}{g^2(b(\quant))}||\bhat(\quant)-b(\quant)|.
$$
Note that $b'(\quant)=1/g(b(\quant))$. Also note that if we differentiate both
sides of this expression with respect to $\quant$, we obtain
$$
b^{\prime\prime}(\quant)=g'(b(\quant))b'(\quant)/g^2(b(\quant)).
$$
By our assumption, $E[(|\ghat(\bhat(\quant))-g(\bhat(\quant)))^2]^{1/2}=O(1/\rate)$.
Also, by Lemma \ref{error bid function}, $|\bhat(\quant)-b(\quant)|=O(b'(q)/(2\sqrt{\samples}))$.
Combining these results, we obtain that 
$$
E[(\bhat'(\quant)-b'(\quant))^2]^{1/2}=O\left(
\frac{b'(\quant)^2}{\rate}+\frac{|b^{\prime\prime}(\quant)|}{2 \sqrt{\samples}}
\right).
$$
\end{proofof}

\begin{proofof}{Theorem \ref{theorem1}}
The proof of this theorem reduces to substitution of appropriate
expressions of $b'(\cdot)$ and $b^{\prime\prime}(\cdot)$
into Theorem \ref{th:derivative}.
For all-pay auctions $b'(\quant)=\alloc'(\quant)v(\quant)$
and $b^{\prime\prime}(\quant)=\alloc^{\prime\prime}(\quant)v(\quant)
\alloc'(\quant)v'(\quant)$.
Therefore, we can express
$$
\frac{MSE_{R(\quant)}(N)}{R(\quant)}=O\left(
\frac{\alloc'(\quant)v(\quant)}{\rate}
+\frac{|\alloc^{\prime\prime}(\quant)v(\quant)+
\alloc'(\quant)v'(\quant)|}{2\sqrt{\samples}\alloc'(\quant)v(\quant)}
\right).
$$
We guarantee the bound $\epsilon(\quant)$ for this expression
if each term is bounded by $\epsilon(\quant)$:
$$
\frac{\alloc'(\quant)v(\quant)}{\rate},\,
\frac{|\alloc^{\prime\prime}(\quant)|}{2\sqrt{\samples}\alloc'(\quant)},\,
\frac{v'(\quant)}{2\sqrt{\samples}v(\quant)} \leq \epsilon(\quant).
$$

\end{proofof}

\begin{proofof}{Theorem \ref{theorem2}}
  Following the analysis of Theorem~\ref{theorem1}, 
we substitute the appropriate
expressions of $b'(\cdot)$ and $b^{\prime\prime}(\cdot)$
into Theorem \ref{th:derivative}.
Note that for the first-price auction
$$
b'(\quant)=(v(q)-b(q))\frac{\alloc'(\quant)}{\alloc(\quant)}
$$
and 
$$
b^{\prime\prime}(\quant)=v'(\quant)\frac{\alloc'(\quant)}{\alloc(\quant)}
+(v(\quant)-b(\quant))\frac{\alloc^{\prime\prime}(\quant)}{\alloc(\quant)}
-2(v(\quant)-b(\quant))\left(\frac{\alloc^{\prime}(\quant)}{\alloc(\quant)}\right)^2.
$$
This means that we can evaluate
\begin{align*}
\frac{MSE_{R(\quant)}(N)}{R(\quant)}&=O\bigg(
\frac{(v(q)-b(q))}{\rate}\frac{\alloc'(\quant)}{\alloc(\quant)}\\
&+\frac{\left|
v'(\quant)\frac{\alloc'(\quant)}{\alloc(\quant)}
+(v(\quant)-b(\quant))\frac{\alloc^{\prime\prime}(\quant)}{\alloc(\quant)}
-2(v(\quant)-b(\quant))\left(\frac{\alloc^{\prime}(\quant)}{\alloc(\quant)}\right)^2
\right|
}{2 \sqrt{\samples}\frac{(v(q)-b(q))}{\rate}\frac{\alloc'(\quant)}{\alloc(\quant)}} \bigg).
\end{align*}
To guarantee the bound $\epsilon(\quant)$, we need to bound
each of the terms by $\epsilon(\quant)$. Thus, we require that
$$
\frac{(v(q)-b(q))}{\rate}\frac{\alloc'(\quant)}{\alloc(\quant)},\,
\frac{v'(q)}{(v(q)-b(q))2\sqrt{\samples}},\,
\frac{|\alloc^{\prime\prime}(\quant)|}{\alloc'(\quant)2\sqrt{\samples}},\,
\frac{\alloc'(\quant)}{\sqrt{\samples}\alloc(\quant)} \leq \epsilon(\quant).
$$

\end{proofof}

\end{document}